\documentclass[journal]{IEEEtran}

\usepackage[letterpaper,margin=1.45cm,footskip=0.4cm]{geometry}
\usepackage{amsthm}
\usepackage{algorithm,amsbsy,amsmath,amssymb,epsfig,mathrsfs,url,color,cite}
\usepackage{mathtools}
\usepackage{bm}

\usepackage{algorithm, algorithmic}
\usepackage{graphicx,epstopdf,tikz,pgfplots,amssymb,amsfonts,amsmath,mathrsfs,lettrine}
\usepackage{url}
\usepackage{geometry}
\usepackage{enumerate}
\usepackage{array}

\usepackage{graphicx,xcolor}
\usepackage{multirow}
\usepackage{color, colortbl}

\newcommand{\beq}{\begin{equation}}
\newcommand{\eeq}{\end{equation}}
\newcommand{\beqn}{\begin{eqnarray}}
\newcommand{\eeqn}{\end{eqnarray}}

\newtheorem{theorem}{\textbf{Theorem}}
\newtheorem{lem}{\textbf{Lemma}}

\newcommand{\eref}[1]{(\ref{#1})}
\newcommand{\sref}[1]{Section~\ref{#1}}

\newcommand{\ignore}[1]{}
\newcommand{\thref}[1]{Theorem~\ref{#1}}

\usepackage[labelformat=simple]{subcaption}

\long\def\symbolfootnote[#1]#2{\begingroup%
\def\thefootnote{\fnsymbol{footnote}}\footnote[#1]{#2}\endgroup}

\usepackage{epstopdf}
\newcount\colveccount
\newcommand*\colvec[1]{
\global\colveccount#1
\begin{pmatrix}
\colvecnext
}
\def\t#1{
#1
\global\advance\colveccount-1
\ifnum\colveccount>0
\\
\expandafter\colvecnext
\else
\end{pmatrix}
\fi
}

\newcolumntype{P}[1]{>{\centering\arraybackslash}p{#1}}

\begin{document}

\title{Throughput Maximization of Network-Coded and Multi-Level Cache-Enabled Heterogeneous Network}

\author{Mohammed S. Al-Abiad, Student Member, IEEE, Md. Zoheb Hassan, Student Member, IEEE, and Md. Jahangir Hossain, Senior Member, IEEE


\thanks{Mohammed S. Al-Abiad, Md. Zoheb Hassan, and Md. Jahangir Hossain are with the School of
	Engineering, University of British Columbia, Kelowna, BC V1V 1V7, Canada
	(e-mail: m.saif@alumni.ubc.ca, zohassan@mail.ubc.ca, jahangir.hossain@ubc.ca).
	
}
\vspace{-.7cm}}
\maketitle

\begin{abstract}
One of the paramount advantages of multi-level cache-enabled (MLCE) networks is pushing contents proximity to the network edge and proactively caching them at multiple transmitters (i.e., small base-stations (SBSs), unmanned
aerial vehicles (UAVs), and cache-enabled device-to-device (CE-D2D) users). As such, the fronthaul congestion  between a core network and a large number of
transmitters is alleviated. For this objective, we exploit network coding (NC) to schedule a set of users to the same transmitter. Focusing on this, we consider the throughput maximization problem that  optimizes jointly   the network-coded user scheduling and  power allocation, subject to fronthaul capacity, transmit
power, and NC constraints. Given the intractability
of the problem, we decouple it into two separate subproblems. In the first subproblem, we consider the network-coded user scheduling problem for the given power allocation, while in the second subproblem, we use the NC resulting user schedule to optimize the power levels. We design an innovative \textit{two-layered rate-aware NC (RA-IDNC)} graph to solve the first subproblem  and evaluate the second subproblem using an iterative function
evaluation (IFE) approach. Simulation results are presented to depict
the  throughput gain of the proposed approach over the existing solutions.
\end{abstract}

\begin{IEEEkeywords}
Cache-enabled networks, file streaming, multi-level caching, network coding,  power allocation.
\vspace{-.4cm}
\end{IEEEkeywords}

\IEEEpeerreviewmaketitle

\section{Introduction}

\IEEEPARstart{T}{he} fifth generation (5G) wireless communication is expected to alleviate the data congestion at the core network and improve system throughput. In fact, 5G and beyond will be revolutionary in advancing the current system capabilities in terms of connectivity, intelligence, and user demand availability \cite{1}. The popularity of user demand, where users request duplicate copies of same popular contents within a short period of time, leads to substantial fronthaul congestion. To this end, developing \textit{caching} techniques is crucial. By caching the popular content in intermediate nodes, user demand can be delivered easily, without duplicate transmissions from cloud storage servers. Hence, redundant requests can be tremendously eliminated \cite{2}.

Edge caching is a promising technique for
significantly alleviating data congestion  and improving throughput performance \cite{3, 4, 5, 6}. For example, a throughput-physical layer design with one-to-one user-cache connections has been considered in a number of works, e.g., \cite{ 5, 6, 6nn}. The authors of \cite{6nn} maximized sum-rate in multi-level cache-enabled networks, subject to fronthaul capacity, however, they viewed the network solely from the physical-layer perspective without taking into consideration upper-layer facts, e.g., combining users' demands. As a result, in \cite{6nn},
only a single user was assigned to each transmitter, and it is inefficient. For the cache-enabled network, such a limitation  can be overcome by applying the  instantly decodable network coding (IDNC) \cite{7,8}.  As the name indicates, IDNC is an instantaneous coding technique that is suitable for popular contents real-time delivery \cite{6n}. Further, the simplicity of encoding and decoding of contents using XOR binary operation  makes IDNC appropriate for the low-complexity devices \cite{7, 8, 8nn, 9, 10}.

Several NC-enabled caching techniques that addressed the throughput maximization problem are available in literature as follows. In \cite{9}, the authors developed a framework for IDNC transmission time optimization from multiple storage servers. Based on the model in \cite{9}, the authors of \cite{10} employed IDNC
in fronthaul offloading, subject to the immediate content delivery constraint. These works are limited in the sense that they focused on network-layer view of the problem and ignored physical-layer aspects, such as rate and power. To overcome this limitation, the works in \cite{7, 8} included the transmitting rate and power in their network-layer optimizations. However, these works assumed that all the requested contents are fully cached and stored at the network edge. More importantly, they did not consider multi-level caches and only considered a vanilla version of throughput maximization problem by allocating users to storage servers and optimize their powers. In this paper, we consider the
scheduling of content delivery with the objective of maximizing system throughput in a coded multi-level cache-enabled (MLCE) network, subject to fronthaul capacity, caching, power limitation, and NC.

Motivated by the above-mentioned limitations of earlier
NC works, in this paper we consider a coded three-level cache-enabled network, wherein cache capabilities are enabled at SBSs, UAVs, and cache-enabled D2D (CE-D2D) users. We solve the joint optimization of NC user scheduling and power allocation problem by maximizing the network sum-throughput, i.e., maximizing the number of received bits. For such a difficult mixed-integer non-convex optimization problem, we develop an innovative cross-layer NC (CLNC) iterative approach. In particular, our developed approach efficiently iterates between finding the maximum-weight clique (MWC) in the designed two-layered RA-IDNC graph and optimizing the power levels of the transmitters using an iterative function evaluation (IFE) method. Simulation results are presented to depict
the  throughput gain of the proposed approach over the existing solutions.
\ignore{
This paper is organized as follows. In \sref{SMM}, we describe the system model and formulates the throughput maximization problem. In \sref{NC}, the joint power allocation and network-coded user scheduling problem in multi-level
cache-enabled networks is decoupled and solved. In \sref{NR}, presents and discusses the numerical results and in \sref{C}, we conclude the paper.}

\section{System Model and Problem Formulation} \label{SMM}

\subsection{System Model and Parameters} \label{SM}
Similar to \cite{6nn}, we consider a heterogeneous MLCE network, which
consists of $N$ cache-enabled D2Ds (CE-D2Ds) users, $K$ small base-stations (SBSs), and $U$ unmanned aerial vehicles (UAVs). The sets of CE-D2D users, SBSs, and UAVs are denoted
by  $\mathcal{N}=\{1,2,\cdots,N\}$,  $\mathcal{K}=\{1,2,\cdots,K\}$, and  $\mathcal{U}=\{1,2,\cdots,U\}$, respectively. The set of transmitters in the network is denoted by $\mathcal T=\mathcal N \cup \mathcal K \cup \mathcal U $ and they cooperate with each other in streaming a frame of popular files $\mathcal{F} =\{1,2,\ \cdots, F\}$ of $F$ source files to single-antenna $M$ users, denoted by the set $\mathcal{M}=\{1,2, \cdots,M\}$. Fig. \ref{fig1} shows
a three-level cache-enabled network with $2$ SBSs, $2$ UAVs, $1$
CE-D2D user, and $6$ users. 

\begin{figure}[t]
	\centering
	\includegraphics[width=0.45\linewidth]{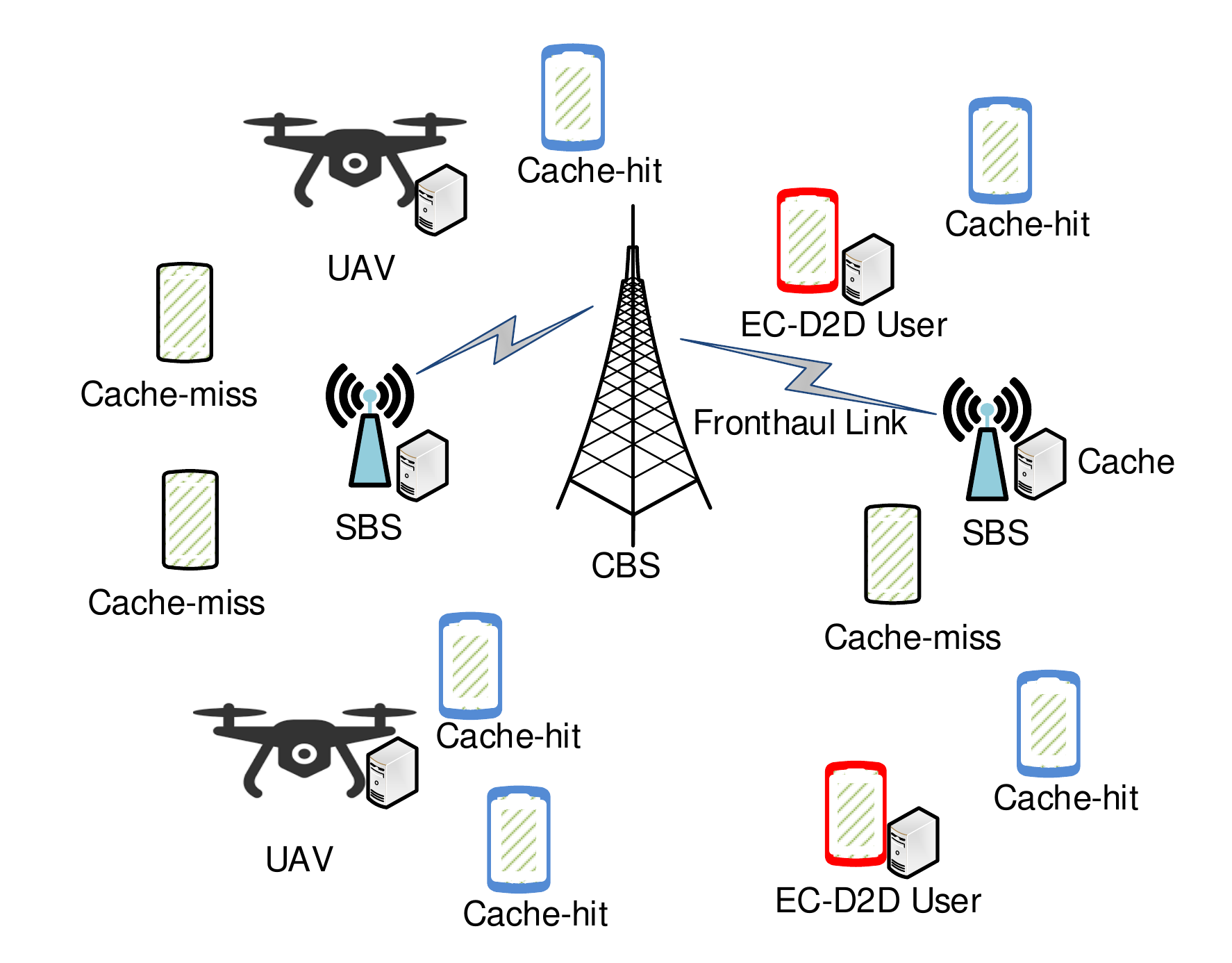}
	\caption{ A heterogeneous MLCE network consists of  $2$
		SBSs, $2$ UAVs, $2$ CE-D2D user, and $8$ users.}
	\label{fig1}
\end{figure}

The set of users $\mathcal M$ can be divided into cache-hit users denoted by the set $\mathcal C_1$ and cache-miss users denoted by the set $\mathcal C_2$. Specifically, $\mathcal C_1$ represents a set of users whose demands are being cached by UAVs and CE-D2D users and $\mathcal C_2$ represents a set of users whose demands are being missed at the network. We consider that only SBSs are
connected to the cloud base-station (CBS), and accordingly, SBSs can pre-fetch missed users' demands from the CBS and serve cache-miss users. The UAVs and CE-D2D users are cooperating in serving cache-hit users. Users' demands usually represent a frame of popular files that is available at the CBS. The CE-D2D users, UAVs, and SBSs can only store $\mu F$ files, where $0\leq\mu\leq 1$ is the caching ratio. The set of cached files by CE-D2D user $d$, UAV $u$, and SBS $s$ are denoted as $\mathcal H_d$, $\mathcal H_u$, and $\mathcal H_s$, respectively. We consider a centralized scenario and therefore, the  CBS is responsible for making the NC scheduling and power allocation decisions. Unless otherwise stated, the term ``transmitter $t$" will be referred to a transmitting CE-D2D user, UAV, or SBS. 

The transmission power level of transmitter $t$ is denoted as $P_{t}$ and thus, the received data rate (in bits/s) at cache-hit user $n$ from  transmitter $t$ can be expressed as $ R_{n,t}=W\log_2(1+\frac{P_{t} \gamma_{n,t}}{1+\sum_{m=1,m\neq t}^TP_{m} \hat{\gamma}_{n,m}})$, where $\gamma_{n,m}$ is the channel-gain-to-noise-ratio (CNR) between the $t$-th transmitter and the $n$-th cache-hit user, and $\hat{\gamma}_{n,m}$ is the interfering-channel-gain-to-noise-ratio (INR) between transmitter $m$ at cache-hit user $n$. For a cache-miss user $n$ that is associated with $(n,s) \in \mathcal C_2$, the files
needs to be fetched from the CBS and served by the
SBS $s$.  In this case, the communication is done in two hops:
1) the CBS sends the data to SBS $s$ through
a fronthaul link of capacity $C_{fh}$, and 2) SBS $s$ transmits the data to the user with an SINR ${\gamma}_{n,s}$, which gives a rate expression similar to $R_{n,t}$. Hence, the achievable overall data rate of a cache-miss
user $n$ associated with SBS $s$ is expressed as $
R_{n,s}=\min\left(C_{fh}, W\log_2\left(1+\frac{P_{s} \gamma_{n,s}}{1+\sum_{m=1,m\neq s}^TP_{m} \hat{\gamma}_{n,m}}\right)\right)$.
We consider that the reception of an uncoded/encoded ﬁle transmitted from transmitter $t$ is successful at user $n$ if $R_{t}\leq R_{n,t}$. The set of achievable capacities of all users across all transmitters is denoted by $\mathcal R$ and represented as a product of the set of the achievable capacities, i.e., $
\mathcal{R} = \bigotimes_{(n,t) \in \ \mathcal{M} \times \mathcal{T} } R_{n,t}$.

We consider that users in the sets $\mathcal C_1$, $\mathcal C_2$  have downloaded some files and are interested to download one file from $\mathcal F$ at any arbitrary time epoch. This can be represented by the side information of user $n$ as follows: 1) the \textit{Has} set $\mathcal{H}_{n}$ that represents previously obtained files by the user $n$; 2) the \textit{Lacks} set $\mathcal {L}_n = \mathcal{F}\backslash \mathcal {H}_n$ that represents non-requested files, 3)
the requested file $\mathcal {W}_n$ by user $n$ in the current scheduling frame.

\subsection{Problem Formulation}
The problem of streaming a set of accessible files from the set of transmitters $\mathcal T$ that maximizes throughput is equivalent to schedule users to the transmitting devices $\mathcal T$ and optimize the power levels under the constraint that each user can stream its requesting files from only one transmitter.

Let $\kappa_{t}$ be a coding combination of some files in $\mathcal F$ that is provided by transmitter $t$ to a set of targeted users denoted by $\tau_{t}$ and expressed as $\tau_{t}= \{n \in \mathcal{N}\ ||\kappa_{t} \cap \mathcal{W}_n| = 1 \mbox{ \& } R_{t}\leqslant R_{n,t}\}$, where $R_{t}$ is the minimum adopted transmission rate of transmitter $t$. Let the binary variable $X_{n,t}$ (where $n \in \mathcal N$ and $t \in \mathcal T$) be $1$ if user $n$ is assigned to transmitter $t$, and $0$ otherwise. Let $\mathcal S$ denote a set of transmitters that schedule users in $\mathcal C_2$. For a typical optimization formulation, let $\mathcal C_3$ be the set of cache-miss users that do not satisfy the feasible condition of being scheduled to a UAV or CE-D2D user. This case is typically occurred when the UAVs and the CE-D2D users do  not
cache the requested files of cache-miss users.  The throughput maximization optimization problem considered in this work can then be formulated as follows
\begin{equation*}
\begin{split}
\mathcal{P}_1: 
\max_{\substack{X_{i,j} \in \{0,1\}\\ P_t, \tau_t, \tau_s}} &\sum _{\substack{t\in \mathcal{T}  \\ n\in \tau_t \cap \mathcal{C}_1 }} X_{n,t} \min_{n \in \mathcal{\tau}_{t}} R_{n,t}+ \sum _{\substack{ s \in \mathcal{S}\\n\in \tau_t \cap \mathcal{C}_2  }} X_{n,s} \min_{n \in \mathcal{\tau}_{s}} R_{n,s}\\
{\rm s.t.\ }
& \text{(C1):}\hspace{0.1cm} \sum _{t} X_{n,t}\leqslant1,~\forall n\in \mathcal{N}, \\
& \text{(C2):}\hspace{0.1cm} \kappa_{t}\subseteq \mathcal P(\mathcal{H}_{t}),~\forall t\in \mathcal{T},\\
& \text{(C3):}\hspace{0.1cm} \sum _{(n,t)\in \mathcal C_3} X_{n,t}=0,\\
& \text{(C4):}\hspace{0.1cm} 0 \leq P_{t} \leq P_{\max},
\end{split}
\end{equation*}
where the optimization is over the assignment variable $X_{n,t}$, NC variables $\{\kappa_{t}, \kappa_{s}\}$, and continuous variable $P_{t}$. Constraint (C1)  ensures assigning the same user to only one transmitter; constraint (C2)  translates the caching limitation of each transmitter; constraint (C3) corresponds to the fact that cache-miss users can only be served by SBSs after fetching the
request from the CBS, and constraint (C4) bounds the maximum transmit power of each transmitter.

($\mathcal P_1$) is a mixed-integer non-linear programming problem. In fact, assuming the given NC scheduling,  a special case of ($\mathcal P_1$) is  non-convex power allocation problem which is NP-hard. Thereby, $\mathcal P_1$ is NP-hard and it can be optimally solved by using branch-and-bound (BB) or exhaustive-search methods. However, these approaches   require high computational complexity that increases exponentially with the number of transmitters and users.
In this work, capitalizing the graph theory, we propose an iterative method to solve ($\mathcal P_1$) with the reduced computational complexity.


\section{Network Coding Scheduling and Power Optimization}\label{I}
In this section, we recommend that instead of solving ($\mathcal P_1$) optimally
due to the prohibitive computational complexity, a better way to
tackle such high complexity is to solve ($\mathcal P_1$) iteratively.
Specifically, ($\mathcal P_1$) is decomposed into the coordinated scheduling and power allocation sub-problems. For the given  transmit power  allocation of each transmitter $t \in \mathcal T$, the coordinated scheduling problem is solved using graph-based method. Afterward, for the resulting NC and user-transmitter pairs schedule, the power allocation problem is solved.  We iterate between these two steps until convergence.

Towards that goal, we first consider a sub-problem of coordinated scheduling,
and it can be written as 
\begin{equation*}
\begin{split}
\mathcal{P}_2: 
\max_{\substack{X_{i,j} \in \{0,1\}\\ \tau_t, \tau_s}} &\sum _{\substack{t\in \mathcal{T}\\n\in \tau_t \cap \mathcal{C}_1 }} X_{n,t} \min_{n \in \mathcal{\tau}_{t}} R_{n,t}+ \sum _{\substack{ s \in \mathcal{S}\\n\in \tau_t \cap \mathcal{C}_2  }} X_{n,s} \min_{n \in \mathcal{\tau}_{s}} R_{n,s}\\
{\rm s.t.\ }
& \text{(C1)},\text{(C2)},\text{(C3)}.
\end{split}
\end{equation*}
On the other hand, for the NC and cache-hit/miss user-transmitter pairs scheduling, obtained from $(\mathcal{P}_2)$,  the power allocation sub-problem is written as
\begin{equation*}
\begin{split}
\mathcal{P}_3: 
\max_{\substack{P_t}} &\sum _{\substack{t\in \mathcal T }} |\tau_t| * \min_{n \in \mathcal{\tau}_{t}} R_{n,t}+ \sum _{\substack{s\in \mathcal S  }} |\tau_s| * \min_{n \in \mathcal{\tau}_{s}} R_{n,s}\\
{\rm s.t.\ }
& 0 \leq P_{t} \leq P_{\max},
\end{split}
\end{equation*}
 In the following two sub-sections, we will solve ($\mathcal{P}_2$), ($\mathcal{P}_3$), respectively, and the overall proposed algorithm will be provided in the last-subsection. 

\subsection{Coordinated Scheduling: Two-Layered RA-IDNC Graph}\label{CS}
Let $\mathcal{A}_k$ denote the set of all possible associations between users, files that are cached by transmitter $t$, and achievable capacities of transmitter $t$, i.e., $\mathcal{A}_k=\mathcal{U}\times \mathcal{H}_t\times \mathcal R$, and $a$ is an element
in  $\mathcal{A}_t$. For convenience, $n_a$ represents user $n$ in  association $a$. The whole set of all possible associations is simply the union of the possible associations of all transmitters, i.e., $\mathcal{A}= \bigcup\limits_{t \in \mathcal{T}}\mathcal{A}_t$.

The \textit{first-layer} RA-IDNC graph is denoted by $\mathcal G_1(\mathcal V_1, \mathcal E_1)$  wherein $\mathcal V_1$
and $\mathcal E_1$ are the set of vertices and edges, respectively. This graph is introduced to care about all the users whose requested files are being cached by the set of transmitters $\mathcal T$. A vertex $v \in \mathcal{V}_1$ in this graph is generated for each association in $\mathcal A$, i.e., $v=a$ and $|\mathcal V_1|=|\mathcal A|$. The set of edges $\mathcal E$ is constructed by connecting each two distinct vertices $v_{nf}$ and
$v_{n'f'}$ representing the same transmitter $t$ if one of the two following IDNC conditions is true:
\begin{itemize}
	\item \textbf{C1:} $f_a=f_{a'} \Rightarrow $ File $f$ is needed by users $n_a$ and $n_{a'}$.
	\item \textbf{C2:} $f \in \mathcal H_{a'}$ and $f_{a'} \in \mathcal H_a \Rightarrow$ The file combination $f_a\oplus f_{a'}$ is	instantly decodable for both users $n_a$ and $n_{a'}$.
\end{itemize}
Moreover, two distinct vertices $v_{nf}$ and
$v_{n'f'}$ representing the same transmitter $t$ are connected to each other if their adopted rates are equal, i.e., \textbf{C3:} $R_a=R_{a'}$. Similarly, two distinct vertices $v_{nf}$ and
$v_{n'f'}$ representing two different transmitters $t$, $t'$ are connected to each other if this condition is satisfied \textbf{C4:} $n_a = n_{a'}$ and $t=t'$. \textbf{C4} insists that
the same user cannot be scheduled to more than one transmitter.

The above construction of graph $\mathcal G_{1}$ includes only the requested files that are cached by transmitter $t$, $\forall t \in \mathcal T$. To this end, we suggest designing a second graph that contains vertices for all the users that are not in the IDNC schedule of $\mathcal G_{1}(\mathcal V_1, \mathcal E_1)$, so as to possibly schedule users in the sets $\mathcal C_3$, $\mathcal C_2$ to the available SBSs. Specifically, after choosing the combination
$\kappa_t(\mathcal T)$, $\forall t\in \mathcal T$, we 
propose including vertices for un-scheduled users in $\mathcal G_{1}$ as long as they do not compromise
the feasibility conditions of the scheduled users
in $\tau(\kappa_t), \forall t \in \mathcal T$. This is done by designing the \textit{second-layer} RA-IDNC graph, denoted by $\mathcal G_{2}(\mathcal V_2, \mathcal E_2)$, that can be
constructed by generating a vertex $v_{nf}$ for all users $n \in \mathcal C_2 \cup \mathcal C_3$. The edges of $\mathcal G_{2}$ are generated by \textbf{C1}, \textbf{C2}, \textbf{C3}, and \textbf{C4} similar to $\mathcal G_{1}(\mathcal V_1, \mathcal E_1)$. Consequently, the union of $\mathcal G_{1}$, $\mathcal G_{2}$ gives the entire graph $\mathcal G$ and its corresponding vertices is $\mathcal{V} = \bigcup\limits_{i=\{1,2\}}\mathcal{V}_{i}$.

The weight of each vertex is the weighted sum-throughput of the represented associations, i.e., the weight of vertex $v$ that is associated with $a$ is
\begin{equation}\label{v}
w(v)=
\begin{cases}
R_{n_a,t} ~~\text{if $v \in \mathcal V_1$},\\
\min\left(C_{fh}, R_{n_a,s}\right) ~~\text{if $v \in \mathcal V_2$}.
\end{cases}
\end{equation}

Using the constructed two-layered IDNC graph,  $\mathcal P_2$ is similar to
MWC problems in several aspects. In MWC, two vertices must be adjacent in the graph,
and similarly, in ($\mathcal P_2$), two different users
must be scheduled to two different transmitters. More importantly, the objective of ($\mathcal P_2$) is to maximize
the total sum-rate of all users, and similarly, the
goal of MWC is to maximize the weight of all the vertexes. Consequently, we have the following theorem. 
\begin{theorem} \label{th:1w}
\textit{Using the two-layered RA-IDNC graph, problem ($\mathcal P_2$) can be equivalently transformed to the problem of determining the MWC in $\mathcal G$.}\label{th1}
\end{theorem}

\begin{proof} 
Let $\mathbf{C}^*=\{v_{1},v_{2}, \cdots, \, v_{|\mathbf C|}\},~ \forall v \in \mathcal G$ be the MWC and  is associated with the feasible schedule $\{ \{a_1,a_2,\cdots, \, a_{|\tau_{1}|}\},\ \cdots, \, \{a'_1,a'_2,\cdots, \, a'_{|\tau_{|\mathcal T|}|}\} \}$. Let $\mathbf C$ is the set of all possible cliques in $\mathcal G$. For each vertex $v \in \mathbf C^*$ that is associated with association of user, file, transmitter, and rate, the weight $w(v)$ in (2) is the maximum possible data rate that the induced user in $v$ receives, i.e., the utility under the optimal transmit power. Therefore, the weight of the MWC $\mathbf{C}^*$  is precisely the objective function of problem ($\mathcal P_2$) and can be written as $ w(\mathbf{C}^*)= \sum\limits _{v\in \mathbf{C}^*}w(v)= \sum\limits _{a\in \mathcal{A}}w(a) =\sum\limits _{a\in \mathcal A}\min_{n_a \in \mathcal{\tau}_{t}} R_{n,t} +\sum _{a \in \mathcal A} \min_{n_a \in \mathcal{\tau}_{t}} R_{n,t}.$  Since two vertices are connected, i.e., different users are scheduled to different transmitters, constraint (C1) holds. Moreover, all the vertices are indeed generated based on the cached files of each transmitter, and accordingly, constraint (C2) holds. In addition, conditions \textbf{C1}, \textbf{C2} ensure the NC combinations of constraint (C2).
	
\end{proof}

\subsection{Power Control Optimization}\label{PC}
In this sub-section, we derive optimal power allocations
to maximize sum-throughput for a given MWC that is associated with network-coded user scheduling. Specifically, for a fixed network-coded user scheduling, the non-convex problem ($\mathcal P_3$) can be solved 
locally  using IFE approach. 
The main strategy of such approach is to let the gradient of the objective function
is equal to zero and to
manipulate the stationarity of Karush-khun-Tucker (KKT)
conditions. As such, we derive the proper power update manipulation that is shown in the following lemma. 

\begin{lem} \label{pr:1}
Given the resulting NC and user-transmitter schedule in
problem ($\mathcal P_3$), a converged power
allocation is obtained by updating power $\{P_t\}_{t\in \mathcal T}$ at the $(k + 1)$-th iteration, based on the following power update manipulation
\begin{eqnarray}
\label{Power_update}
P_t^{k+1}=\left[\frac{|\tau_t|*\frac{\text{SINR}_{\hat{n},t}}{1+\text{SINR}_{\hat{n},t}}}{\sum_{\substack{{l=1, l\neq t}}} \left(|\tau_l|*\frac{(\text{SINR}_{\hat{u},l})^2}{1+\text{SINR}_{\hat{u},l}}\right)\frac{\gamma_{\hat{u},l}}{{P}_l\gamma_{\hat{u},l}}}\right]_{0}^{P_{\max}},
\end{eqnarray}
where  $\hat{n}=\arg \min_{n \in \mathcal{\tau}_{t}}R_{n,t}$,  $\hat{u}=\arg \min_{n \in \mathcal{\tau}_{l}}R_{u,l}$, $\forall t, l \in \mathcal{T}$ and $t \neq l$;  $\text{SINR}_{\hat{n},t}=\frac{P_{t}^{(k)} \gamma_{\hat{n},t}}{1+\sum_{m=1,m\neq t}^TP_{m}^{(k)} \hat{\gamma}_{\hat{n},m}}$; and 
$P_t^{(k)}$ is the transmit power of the $t$-th transmitting node at the $k$-th iteration.
\end{lem}
\begin{proof}
Leveraging  the MWC solution to problem ($\mathcal P_2$),  problem ($\mathcal P_3$) can be re-written as follows
\begin{equation} \label{4eq}
\begin{split}
\max_{\substack{P_t}}& \sum _{t \in \mathcal T } |\tau_t| * \min_{n \in \mathcal{\tau}_{t}} R_{n,t}\\ 
{\rm s.t.\ }
& 0 \leq P_{t} \leq P_{\max}.
\end{split}
\end{equation}
Now, \textit{Lemma} \ref{pr:1} can be proved in two main steps that are provided as follows. First, we express the problem \eref{4eq} for $P_t$ in terms of all the other transmitters $\{P_l\}_{l \neq t, l \in \mathcal T}$ and take the partial derivative of the objective function of the power allocation in \eref{4eq} with respect to $P_t$. The partial derivative of the
objective function of the power allocation problem \eref{4eq} with
respect to $P_t$ is as follows
\begin{align}
\label{power_opt_2}
& \frac{\partial R_t}{\partial P_{t}}= \frac{\partial }{\partial P_{t}}\left(|\tau_t| *\min_{n \in \mathcal{\tau}_{t}} R_{n,t}+ \sum _{\substack{{l=1, l\neq t}}} |\tau_l| * \min_{u \in \mathcal{\tau}_{l}} R_{l,u} \right)\\
&=\frac{1}{P_t} \frac{|\tau_t|*\text{SINR}_{\hat{n},t}}{1+\text{SINR}_{\hat{n},t}}-\sum_{l=1, l\neq n}  \left(\frac{(\text{SINR}_{\hat{u},l})^2}{1+\text{SINR}_{\hat{u},l}}\right) \frac{|\tau_l|\gamma_{\hat{u},t}}{\widehat{P}_l\gamma_{\hat{u},l}}.
\end{align}	
Second, by solving $\frac{\partial R_t}{\partial P_t}=0$, we obtain	
\begin{eqnarray}
\label{Power_updatef}
P_t=\frac{|\tau_t|*\frac{\text{SINR}_{\hat{n},t}}{1+\text{SINR}_{\hat{n},t}}}{\sum_{\substack{{l=1, l\neq t}}} \left(|\tau_l|*\frac{(\text{SINR}_{\hat{u},l})^2}{1+\text{SINR}_{\hat{u},l}}\right)\frac{\gamma_{\hat{u},l}}{{P}_l\gamma_{\hat{u},l}}}.
\end{eqnarray}
By solving \eref{Power_updatef}, we can obtain the stationary power allocation for the $t$-th transmitter,$  \forall t$.  However, a closed-form power allocation to solve \eref{Power_updatef} is
intractable. To this end, we aim to solve \eqref{Power_updatef} iteratively.  Specifically, we denote $P_t^{(k)}$ is the given transmit power of the $t$-th transmitter at the $k$-th iteration.  By evaluating the R. H. S. of \eqref{Power_updatef} for $P_t^{(k)}$, $\forall t$, and  projecting the evaluated value in the feasible region  $\eqref{4eq}$, we obtain the  transmit power of the $t$-th transmitter at the $(k+1)$-th iteration, $\forall t$. By iteratively repeating the aforementioned two steps, we obtain the converged transmit power allocations for all the transmitting nodes. This completes the proof of \textit{Lemma \ref{pr:1}}. 
\end{proof}
A more detailed analysis of convergence of the aforementioned power allocations can be found in \cite{6nn} and the references therein. Owing to the space limitation, such an analysis is omitted.

\subsection{Overall Algorithm Development}

The  aforementioned two-steps of solving the MWC problem for a given power allocation and optimizing the power allocation for the resulting NC and user schedule that are explained in \sref{CS} and \sref{PC}, respectively, are iterated until we obtain feasible and converged solution to the overall joint optimization problem in ($\mathcal P1$). By integrating theses two steps, we obtain  Algorithm \ref{Alg1}. Algorithm \ref{Alg1} provides a sub-optimal yet efficient solution to original optimization problem $\mathcal{P}_1$ while requiring significant less computational complexity than an exhaustive search method.

The computational complexity of the proposed algorithm depends on the complexity of solving ($\mathcal P_2$), ($\mathcal P_3$). For solving ($\mathcal P_2$), we need a computational complexity of $O(M^2F^2)$ while for solving  ($\mathcal P_3$), we need $C_p=O\left(|\mathtt{S_1}| \times |\mathtt{S_2}| \times \cdots |\mathtt{S_{|\mathcal T|}}|\right)$ where $\mathtt{S_1}$ represents the NC user scheduling of the first transmitter. Therefore, the overall computational complexity of our proposed scheme is $O\left(K(M^2F^2+C_p)\right)$, where $K$ is the number of iterations. Therefor, the computational complexity is reduced significantly compared to an exhaustive search.

\setlength{\textfloatsep}{0pt}
\begin{algorithm}[t!]
	\caption{The proposed CLNC approach for solving problem ($\mathcal P1$).} \label{Alg1}
	\begin{algorithmic}[1]
		\STATE \textbf{Initialization:} Set the initial values for $P_{t}=P_{\max}, \forall t \in \mathcal T$ and set $K=1$.

		\REPEAT
		
		\STATE Solve problem ($\mathcal P2$) to obtain the NC and user-transmitter association according to \sref{CS} and \thref{th:1w}.
		
		\STATE Evaluate ($\mathcal P3$) as in \sref{PC} to obtain the power levels of the corresponding transmitters as follows:
			\FOR{ Set $k=1$}

		\STATE Update $P_t= P_t^k$ according to Lemma 1.
		
		\STATE  Set $k=k+1$ and iterate until convergence.
		
		\ENDFOR	
		\STATE  Set $K=K+1$
		\UNTIL{Convergence.}
		
		\STATE \textbf{Output:} Final transmission power for all the transmitting UDs. 
		
	\end{algorithmic}
	
\end{algorithm}

\section{Numerical Results}\label{NR}
In this numerical results section, we assess the throughput maximization
performance of the proposed algorithm against the following existing schemes.
\begin{itemize}
	\item \textbf{Classical IDNC}: This scheme considers only network-layer optimization. For successful file's reception, the selected
	transmission rate in each transmitter is the minimum rates of all the scheduled users.
	\ignore{\item \textbf{Random Linear Network Coding (RLNC):} This scheme associates users to a single transmitter to which it has the maximum rate. RNC is employed for each transmitter that has multiple associated users. The encoding is done by mixing all files with different random coefficients, and similar to the traditional IDNC, a minimum rate of the associated users is selected for the transmission.}
   \item \textbf{Uncoded:} Similar to \cite{6nn}, only physical-layer optimization is considered and each transmitter schedules only one user.
	\item \textbf{RA-IDNC:} This scheme was studied in \cite{8nn} where a minimum transmission rate is selected for all the transmitters.
\end{itemize}
We consider a downlink three-level cached-enabled system that consists of $3$ SBSs, $2$ UAVs, $3$ CE-D2D users, and
$15$ users. The transmitters (i.e., SBSs, UAVs, CE-D2D users) have fixed locations and the users are distributed randomly within a hexagonal cell of radius $900$m. The considered channel gains for user-transmitter pairs follow the standard path-loss model that has three parts: 1) path-loss of $128.1 + 37.6 \log_{10}(\text{dis.[km]})$dB; 2) log-normal shadowing with $4$dB standard deviation, and 3) Rayleigh channel fading with zero-mean and unit variance. The channels are assumed to be perfectly estimated. The noise power and the maximum power are set to $\sigma^2=-168.60$ dBm/Hz and $P_\text{max}=-42.60$ dBm/Hz, respectively. The bandwidth is $10$MHz. The caching ratio $\mu$ is set to $0.6$.
\ignore{
\begin{figure}[t!]
	\centering
	\includegraphics[width=0.6\linewidth]{./fig1new.eps}
	\caption{Average sum throughput in Mbps vs the number of users $M$.}
	\label{fig2}
\end{figure}}

\begin{figure}[t!]
	\centering
	\begin{subfigure}[t]{0.23\textwidth}
		\centerline{\includegraphics[width=1.12\linewidth]{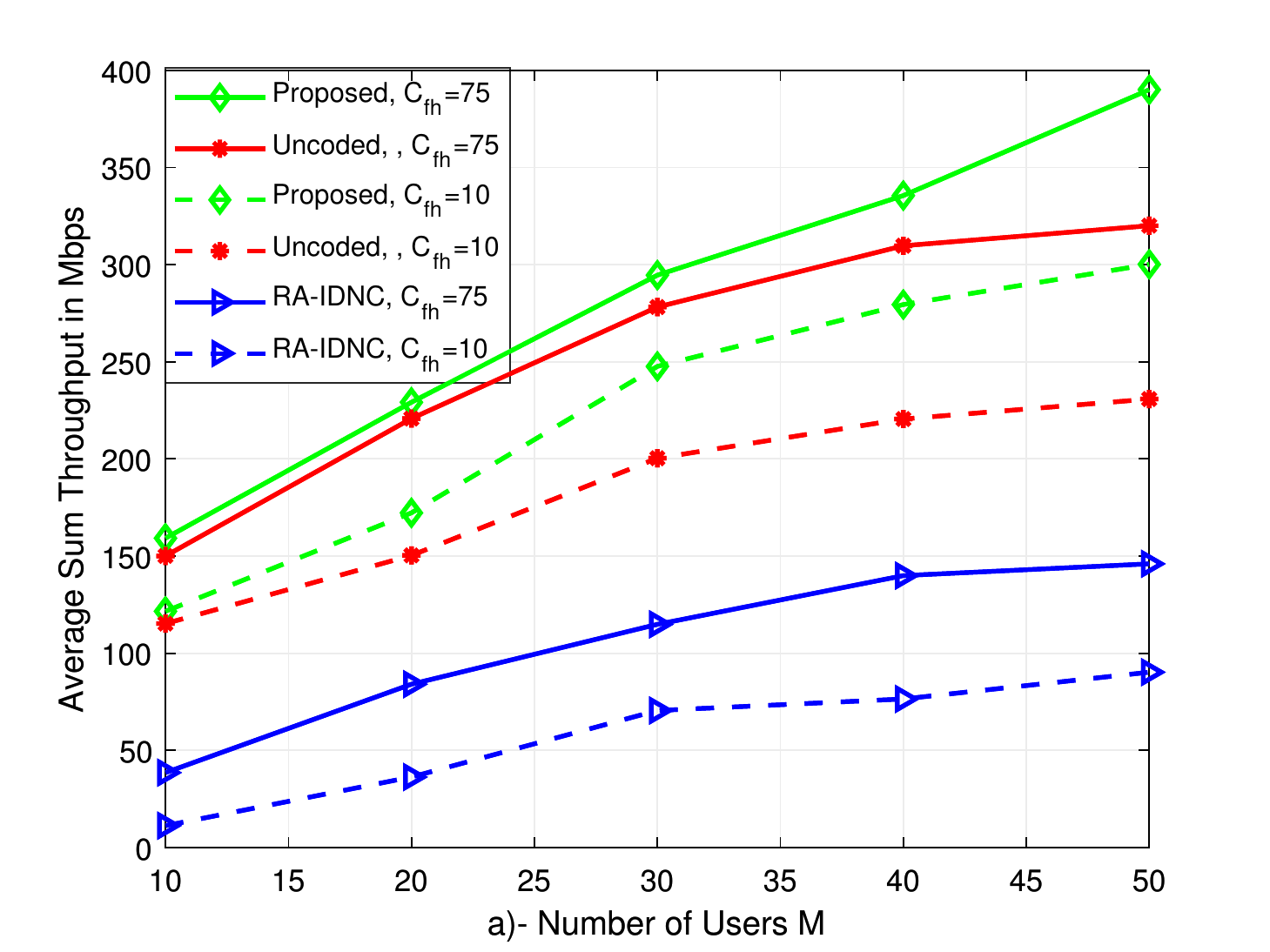}}
		\label{MCT}
	\end{subfigure}%
	~
	\begin{subfigure}[t]{0.23\textwidth}
		\centerline{\includegraphics[width=1.12\linewidth]{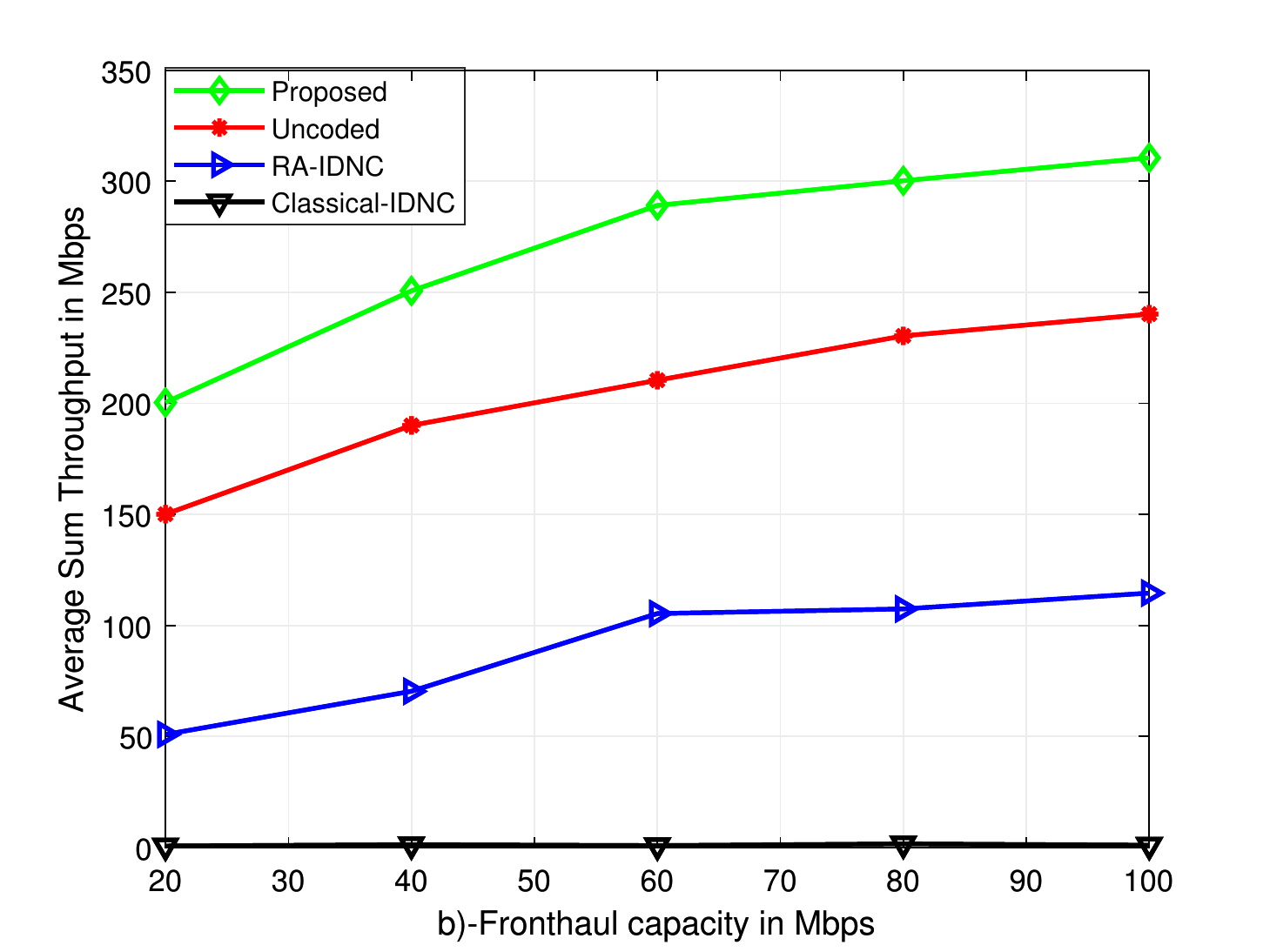}}
		\label{NCT}
	\end{subfigure}
	\caption{Average sum-throughput in Mbps vs: a) number of users $M$; b) fronthaul capacity $C_{fh}$.}
	\label{fig2}
\end{figure}

In Fig. \ref{fig2}-a, we plot the average sum throughput versus the number of users for a three-level cache-enabled network comprising $5$ SBSs, $3$ UAVs, $5$ CE-D2D users, and $30$ files for different values of fronthaul capacity (i.e., $C_{fh}=10, C_{fh}=75$). From the figure, we see that our proposed scheme offers improved performance in terms of throughput maximization with different values of $C_{fh}$ as compared to the other schemes. Such
significant gain is due to the role of our proposed scheme that smartly multiplexes users and selects the best transmission rate of each transmitter, and further controls the power for
interference mitigation.  For a small network, i.e., $M=10$, the uncoded scheme offers close performance to the proposed scheme, i.e., only $10 \%$ improvement over uncoded for $M=10$. This is due to the fact that the proposed scheme suffers from a small number of encoded scheduling opportunities in each transmitter. On the other hand, for a relatively dense network, i.e., $M=50$, the proposed scheme achieves a good performance improvement ($35 \%$) as compared to the contemporary uncoded scheme. This is because the encoded opportunities of multiplexing users to each transmitter increases with the number of users $M$. The performances of the RA-IDNC scheme is heavily suffered from selecting the minimum transmission rate of all the transmitters.

In Fig. \ref{fig2}-b, we show the average sum throughput versus the fronthaul capacity for a three-level cache-enabled network comprising $5$ SBSs, $2$ UAVs, $3$ CE-D2D users, $30$ users, and $30$ files. Similar to what we have discussed in Fig. \ref{fig2}, our proposed scheme shows significant gain on the sum throughput maximization regime as compared to all schemes. From the figure, we can see that our proposed scheme offers an improved performance gain  in the range of $15$\%- $35$\% as compared to the uncocded scheme in  \cite{6nn}.

Finally, Fig. \ref{fig4} shows the average sum throughput of different cache-levels methods of the proposed scheme versus the number of users $M$ for a network comprising
$3$ SBSs, $2$ UAVs, $3$ CE-D2D users, and $40$ files
for different values of fronthaul capacity (i.e., $C_{fh}=10, C_{fh}=75$). The
caching role here becomes particularly pronounced in dens
networks, as shown in Fig. \ref{fig4}, which shows the improved performance of our proposed approach as the number of users increases. Such improved throughput performance demonstrates the pronounced role of our proposed multi-level caching in mitigating the fronthaul
congestion of dense networks. Specifically, at low-fronthaul
capacity regimes, the performance gain of the three-level caching is almost three times the average sum throughput gain of the no-caching
method. Compared with the unocded scheme with two-level caching, our proposed scheme shows a significant improvement in terms of the throughput maximization. Therefore, our proposed scheme provides a more effective way to use MLCE networks and IDNC, and as a result, it is more suitable to be implemented in massive networks for real-time streaming.
\begin{figure}[t!]
	\centering
	\includegraphics[width=0.55\linewidth]{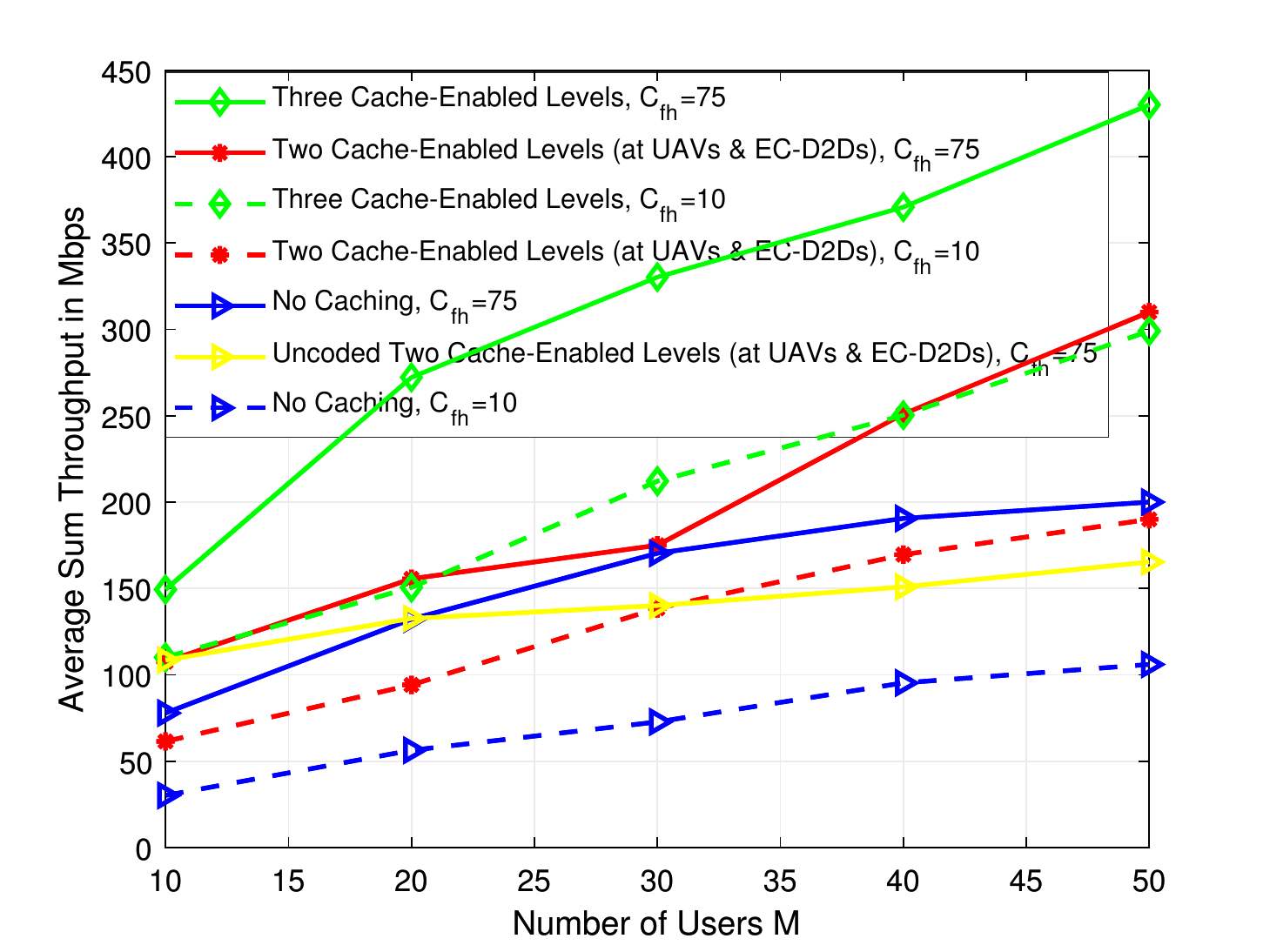}
	\caption{Average sum-throughput in  Mbps vs the number of users $M$.}
	\label{fig4}
\end{figure}

\section{Conclusion}\label{C}
We developed an efficient NC schedule and power optimization framework that improves throughput in MLCE system, where SBSs, UAVs, and CE-D2D users can enable caching technique. The throughput maximization problem was formulated as a joint optimization over NC schedule and power allocation, where the NC schedule sub-problem is proved to be equivalent to a MWC problem, and the power allocation
sub-problem is solved using the iterative function evaluation
approach. Simulation results demonstrated that
our proposed scheme can effectively maximize throughput, especially in dense networks (i.e., the improved average sum throughput  gain
is $35 \%$) as compared to the contemporary uncoded scheme.

\end{document}